\documentclass[journal,10pt]{IEEEtran}
\IEEEoverridecommandlockouts

\usepackage{cite}
\usepackage{indentfirst}
\usepackage{graphicx}
\usepackage{changepage}
\usepackage{stfloats}
\usepackage{amsfonts,amssymb}
\usepackage{bm}
\usepackage{algorithm}
\usepackage{algorithmic}
\usepackage{amsmath}
\usepackage{amsmath,amssymb,amsfonts}
\usepackage{times}
\usepackage{url}
\usepackage{cite}
\usepackage{textcomp}
\usepackage{xcolor}
\usepackage{color}
\usepackage{setspace}
\usepackage{enumitem}
\usepackage{float}
\usepackage{diagbox}
\usepackage[T1]{fontenc}
\usepackage[utf8]{inputenc}
\usepackage{authblk}
\usepackage{threeparttable}
\usepackage{booktabs}
\usepackage{footnote}
\usepackage{graphicx}
\usepackage{pifont}
\usepackage{subfigure}
\usepackage{mathrsfs}

\def\BibTeX{{\rm B\kern-.05em{\sc i\kern-.025em b}\kern-.08em
    T\kern-.1667em\lower.7ex\hbox{E}\kern-.125emX}}

\newtheorem{lemma}{Lemma}
\newtheorem{proposition}{Proposition}

\newenvironment{proof}{{\indent \indent \it Proof:\quad}}{\hfill $\blacksquare$\par}

\begin{document}

\title{UAV-enabled Integrated Sensing and Communication: Tracking Design and Optimization}

\author{\IEEEauthorblockN{Yifan~Jiang, Qingqing~Wu, \textit{Senior Member, IEEE}, Wen~Chen, \textit{Senior Member, IEEE}, and Kaitao~Meng}
\thanks{Yifan~Jiang is with the State Key Laboratory of Internet of Things for Smart City, University of Macau, Macao 999078, China (email: yc27495@umac.mo) and also with Shanghai Jiao Tong University, Shanghai 200240, China. 
Qingqing~Wu and Wen~Chen are with the Department of Electronic Engineering, Shanghai Jiao Tong University, Shanghai 200240, China (e-mail: \{qingqingwu@sjtu.edu.cn; wenchen@sjtu.edu.cn\}).  
Kaitao Meng is with the Department of Electrical and Electronic Engineering, University College London, WC1E 7JE, UK (email: kaitao.meng@ucl.ac.uk).}
}

\maketitle

\begin{abstract}
    Integrated sensing and communications (ISAC) enabled by unmanned aerial vehicles (UAVs) is a promising technology to facilitate target tracking applications. 
    In contrast to conventional UAV-based ISAC system designs that mainly focus on estimating the target position, the target velocity estimation also needs to be considered due to its crucial impacts on link maintenance and real-time response, which requires new designs on resource allocation and tracking scheme. 
    In this paper, we propose an extended Kalman filtering-based tracking scheme for a UAV-enabled ISAC system where a UAV tracks a moving object and also communicates with a device attached to the object. 
    Specifically, a weighted sum of predicted posterior Cramér-Rao bound (PCRB) for object relative position and velocity estimation is minimized by optimizing the UAV trajectory, where an efficient solution is obtained based on the successive convex approximation method. 
    Furthermore, under a special case with the measurement mean square error (MSE), the optimal relative motion state is obtained and proved to keep a fixed elevation angle and zero relative velocity. 
    Numerical results validate that the solution to the predicted PCRB minimization can be approximated by the optimal relative motion state when predicted measurement MSE dominates the predicted PCRBs, as well as the effectiveness of the proposed tracking scheme. 
    Moreover, three interesting trade-offs on system performance resulted from the fixed elevation angle are illustrated. 
\end{abstract}

\begin{IEEEkeywords}
    \vspace{-0.5mm}
    ISAC, UAV, CRB, tracking
\end{IEEEkeywords}

\section{Introduction}
Integrated sensing and communications (ISAC) is expected to merge enhanced sensing services with conventional communication functionalities, including target detection, localization, and tracking \cite{FanLiu-Toward6G}. 
Existing target tracking services in ISAC systems are mainly provided by fixed-position terrestrial ISAC base stations (BSs). 
However, the sensing performance of terrestrial ISAC BSs is limited by the constrained coverage area and the unfavorable channel conditions, such as non line-of-sight (LoS) links. 
To overcome these challenges, the unmanned aerial vehicle (UAV)-enabled ISAC has received increasing research interests thanks to the high maneuverability of UAVs and high probabilities of providing LoS links \cite{KTMeng-2023-TWC-UAV-IPSAC}. 

Currently, some existing works on UAV-enabled tracking systems investigated the designs for maximizing overall system performance, especially trajectory designs of UAV BSs \cite{JunWu-UAVISAC-secure, JunWu-2023-IOT-UAVISAC}. 
In \cite{JunWu-UAVISAC-secure}, the trajectory of a UAV BS simultaneously communicating with a legitimate user and tracking an eavesdropper was designed to maximize the achievable secrecy rate of the user. 
In \cite{JunWu-2023-IOT-UAVISAC}, the trajectories of multiple UAVs tracking a ground moving user were jointly designed to simultaneously maximize the data rate and minimize the predicted posterior Cramér-Rao bound (PCRB) for the estimated user position, which is constituted by the elements of the state prediction mean square error (MSE) matrix and the predicted measurement MSE matrix. 
However, neither of the two aforementioned works considered the predicted PCRB for velocity estimation, which has crucial impacts on the link maintenance (e.g., beam alignment) and system responding speed \cite{FanLiu-Toward6G}.
In fact, most existing ISAC schemes that only focus on estimating the target position are not feasible for highly-reliable target tracking, which thus motivates our new designs.

In this paper, an extended Kalman filtering (EKF)-based tracking scheme is proposed for a UAV-enabled ISAC system, where the UAV trajectory is optimized to minimize a weighted sum of predicted PCRBs for the predicted state variables (i.e., the relative position and the relative velocity). 
This thus generalizes the previous works on UAV-enabled ISAC system designs.
The main contributions of this paper are summarized as follows:
\textbf{1)} An efficient solution is obtained based on the successive convex approximation (SCA) method to optimize the predicted state variables for the weighted sum-predicted PCRB minimization.
\textbf{2)} Under a special case where the state prediction MSE matrix is ignored, the optimal relative motion state is obtained and proved to maintain a fixed elevation angle and zero relative velocity between the UAV and the object.
\textbf{3)} Simulation results verify that when the predicted measurement MSE dominates the predicted PCRBs, the solution for the weighted sum-predicted PCRB minimization can be approximated by the optimal relative motion state obtained under the considered special case, and further illustrate three interesting trade-offs achieved by the fixed elevation angle.

\begin{figure}[!t]
    \centering
    \includegraphics[width=0.48\textwidth]{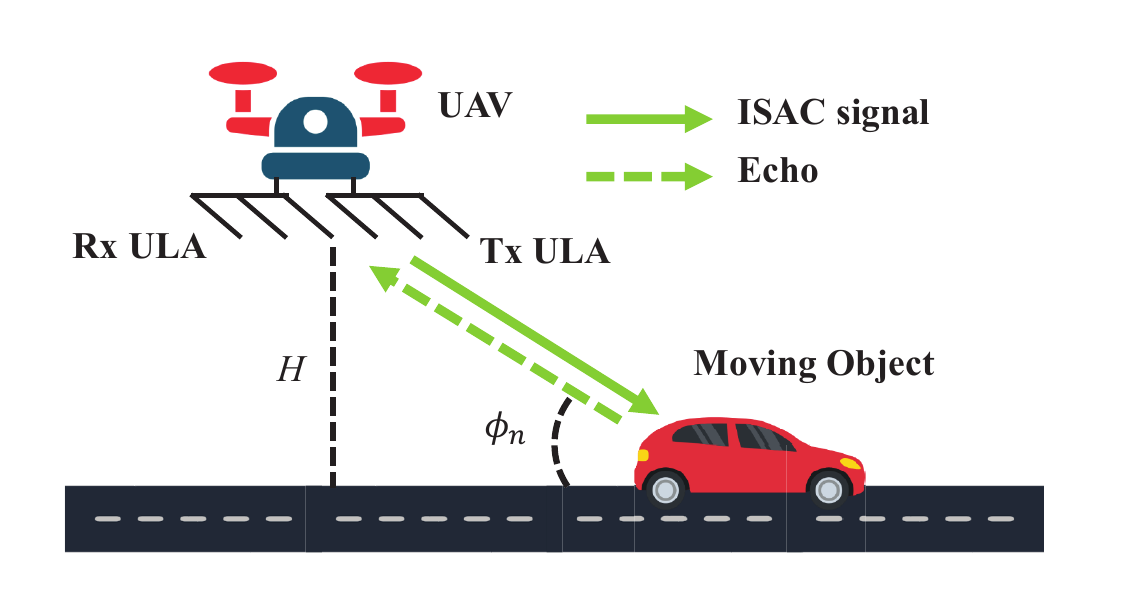}
    \caption{A UAV-enabled ISAC system.}
    \label{fig1}
\end{figure}

\section{System Model and Problem Formulation}\label{SecII}
As illustrated in Fig.\ref{fig1}, we consider a downlink UAV-enabled ISAC system, where a UAV simultaneously tracks a moving object and communicates with a single-antenna device attached to the moving object within a service duration $T$.
Specifically, the whole service duration $T$ is discretized into $N$ small time slots with a constant length of $\Delta T = T/N$. 
Let $n$, $n\in\mathbb{N}=\{1,2,...,N\}$, denote the index of time slots. 
At each time slot, we assume that the motion states (i.e., the position and the velocity) of both the moving object and the UAV keep constant \cite{FanLiu2020TWC-PB}. 
Therefore, at the $n$th time slot, the object motion state relative to the UAV can be represented by state variables $\mathbf{x}_{n}=[x_{n}, v_{n}]^{T}$, where $x_{n}$ and $v_{n}$ denote the relative position and the relative velocity, respectively.
In addition, the UAV is equipped with a uniform linear array (ULA) with $N_{\text{t}}$ transmit antennas along the $x$-axis and a parallel ULA with $N_{\text{r}}$ receive antennas. 
For ease of analysis, we assume that the object moves along a horizontal straight path parallel to the $x$-axis, and the UAV also flies horizontally along the $x$-axis with a constant altitude $H$.\footnote{The UAV and the object trajectories can be extended to the two-dimensional case where both are located in parallel planes by deploying the uniform planar array with receive antennas on the UAV to estimate both the elevation angle and the azimuth angle. Also, it can be extended to the case where the UAV can change its height by modelling the additional Doppler shift resulted from the UAV vertical movement and the UAV height constraint.}

\subsection{EKF-based Tracking Scheme}
In our proposed tracking scheme, the UAV estimates state variables and designs its next time slot's trajectory at each time slot.
Specifically, we assume that the object movement follows the constant-velocity model\cite{CVModel}, and thus the state evolution model can be derived from the geometry as illustrated in Fig.\ref{fig1}, which is given by
\vspace{-1.5mm}
\begin{equation}\label{SEmodel}
	\mathbf{x}_{n} = \mathbf{G}\mathbf{x}_{n-1} - \mathbf{u}_{\text{A},n} + \mathbf{z}_{\text{s},n},
    \vspace{-1.5mm}
\end{equation}
where $\mathbf{G}$ denotes the state transition matrix, $\mathbf{u}_{\text{A},n}=[(v_{\text{A},n} - v_{\text{A},n-1})\Delta T, (v_{\text{A},n} - v_{\text{A},n-1})]^{T}$ denotes the increment of the UAV motion state, $v_{\text{A},n}$ denotes the UAV velocity and $\mathbf{z}_{\text{s},n}$ denotes the zero-mean Gaussian distributed process noise with the covariance matrix $\mathbf{Q}_{\text{s}}$ at the $n$th time slot. 
In (\ref{SEmodel}), the expressions of the state transition matrix $\mathbf{G}$ and the process noise covariance matrix $\mathbf{Q}_{\text{s}}$ can be derived from the constant-velocity model \cite{CVModel} and given by
\vspace{-1.5mm}
\begin{equation}
	\mathbf{G} = \begin{bmatrix}
        1 & \Delta T  \\
        0 & 1 
    \end{bmatrix}, 
    \mathbf{Q}_{\text{s}} = \begin{bmatrix}
        \frac{1}{3}\Delta T^{3} & \frac{1}{2}\Delta T^{2}  \\
        \frac{1}{2}\Delta T^{2} & \Delta T 
    \end{bmatrix}\tilde{q},
    \vspace{-1.5mm}
\end{equation}
respectively, where $\tilde{q}$ is the process noise intensity \cite{CVModel}.
Given the state evolution model, the UAV obtains the estimated state variables for the $n$th time slot given by $\hat{\mathbf{x}}_{n} = \breve{\mathbf{x}}_{n} + \mathbf{K}_{n}(\mathbf{y}_{n} - \mathbf{h}(\breve{\mathbf{x}}_{n}))$, where $\breve{\mathbf{x}}_{n}$ denotes the predicted state variables obtained by\footnote{At the initial time slot, the UAV can estimate the object position and velocity as a mono-static radar.} $\breve{\mathbf{x}}_{n} = \mathbf{G}\hat{\mathbf{x}}_{n-1} - \mathbf{u}_{\text{A},n} = [\breve{x}_{n},\breve{v}_{n}]^{T}$, $\mathbf{K}_{n}$ denotes the Kalman gain matrix given by $\mathbf{K}_{n} = \mathbf{M}_{\text{p},n}\mathbf{H}_{n}^{H}(\mathbf{Q}_{\text{m},n} + \mathbf{H}_{n}\mathbf{M}_{\text{p},n}\mathbf{H}_{n}^{H} )^{-1}$, $\mathbf{M}_{\text{p},n}=\mathbf{G}\mathbf{M}_{n-1}\mathbf{G}^{H} + \mathbf{Q}_{\text{s}}$ denotes the state prediction MSE matrix, $\mathbf{M}_{n-1}$ denotes the estimation MSE matrix at the $(n-1)$th time slot, $\mathbf{y}_{n}$, $\mathbf{h}(\cdot)$, $\mathbf{H}_{n}$ and $\mathbf{Q}_{\text{m},n}$ denote the measurement variables, the measurement model, the Jacobian matrix for $\mathbf{h}(\cdot)$ and the measurement noise covariance matrix at the $n$th time slot, respectively, which will be specified in the following subsection.

\subsection{Signal Model}
\subsubsection{Radar Signal Model}

Let $s_{n}(t)\in\mathbb{C}$ denote the information-bearing signal transmitted to the device at time $t$ in the $n$th time slot. 
Then, the transmitted ISAC signal can be denoted by $\mathbf{s}_{n}(t)=\mathbf{f}_{n}s_{n}(t)$ where $\mathbf{f}_{n}\in\mathbb{C}^{N_{\text{t}}\times 1}$ denotes the transmit beamforming vector in the $n$th time slot. 
The object is modelled as a point scatterer with its radar cross section denoted by $\varepsilon$. 
Since the UAV-object channel is dominated by the direct LoS link \cite{KTMeng-2023-TWC-UAV-IPSAC}, the large-scale radar channel gain follows a free-space path-loss model: $G_{\text{r},n} = \beta_{\text{r}}/d_{n}^{4}$, where $d_{n}$ is the Euclidean distance between the object and the UAV, and the coefficient $\beta_{\text{r}}$ is defined as $\beta_{\text{r}}\triangleq\lambda^{2}\varepsilon/(64\pi^{3})$ with the carrier wavelength denoted by $\lambda$ \cite{largescale}. 
Let $\tau_{n}$ and $\phi_{n}$ denote the round-trip time delay and the elevation angle of the geographical path connecting the UAV to the object at the $n$th time slot, respectively. 
Thus, the radar channel can be expressed as $\mathbf{h}^{\text{R}}_{n} = \sqrt{G_{\text{r},n}}e^{-j4\pi\frac{d_{n}}{\lambda}}\mathbf{b}(\phi_{n})\mathbf{a}^{H}(\phi_{n})$, where $\mathbf{a}(\phi) = [e^{j\frac{\pi(N_{\text{t}}-1)\cos{\phi}}{2}}, ..., e^{-j\frac{\pi(N_{\text{t}}-1)\cos{\phi}}{2}}]^{T}$ and $\mathbf{b}(\phi) = [e^{j\frac{\pi(N_{\text{r}}-1)\cos{\phi}}{2}}, ..., e^{-j\frac{\pi(N_{\text{r}}-1)\cos{\phi}}{2}}]^{T}$.
The reflected echo signals received at the UAV can be given in the form\footnote{The self-interference caused by the simultaneous signal transmission and reception can be cancelled out via radio frequency full-duplex techniques.} $\mathbf{r}_{n}(t) = \sqrt{P_{\text{A}}}\mathbf{h}^{\text{R}}_{n}e^{j2\pi\mu_{n}t}\mathbf{f}_{n}s_{n}(t-\tau_{n}) + \mathbf{z}_{\text{r},n}(t)$, where $P_{\text{A}}$ is the transmit power, $\mu_{n}$ is the Doppler shift and $\mathbf{z}_{\text{r},n}(t)\in\mathbb{C}^{N_{\text{r}}\times 1}$ denotes the complex additive white Gaussian noise with zero mean and variance of $\sigma^{2}$.

\subsubsection{Radar Measurement Model} \label{Subsec::sym-mea}

At the $n$th time slot, the measurement variables are given by $\mathbf{y}_{n}=[\hat{\phi}_{n},\hat{\tau}_{n},\hat{\mu}_{n}]^{T}$, where $\hat{\phi}_{n}$, $\hat{\tau}_{n}$ and $\hat{\mu}_{n}$ denote the measured elevation angle, time delay and Doppler shift, respectively. 
Specifically, $\hat{\phi}_{n}$ can be obtained via the maximum likelihood estimation, and then $\hat{\tau}_{n}$ and $\hat{\mu}_{n}$ can be obtained by performing the standard matched-filtering, which are further expressed as\cite{FanLiu2020TWC-PB}
\vspace{-1.5mm}
\begin{equation}\label{MeasurementModel}
    \left\{ \begin{aligned}
        \hat{\phi}_{n} &= \arctan{(H/x_{n})} + z_{1,n}, \\
        \hat{\tau}_{n} &= 2(x_{n}^{2} + H^2)^{\frac{1}{2}}/c + z_{2,n}, \\
        \hat{\mu}_{n} &= -2f_{c}v_{n}x_{n}/(c(x_{n}^{2} + H^2)^{\frac{1}{2}}) + z_{3,n},
    \end{aligned} 
    \right.
    \vspace{-1.5mm}
\end{equation}
where $c$ is the speed of light, $z_{i,n}\sim\mathbf{\mathcal{N}}\left(0, \sigma_{i,n}^{2}\right), i=1,2,3$ denote the measurement noise for the elevation angle, the time delay and the Doppler frequency, respectively.
Further, the corresponding measurement noise variance $\sigma_{1,n}^{2}$, $\sigma_{2,n}^{2}$ and $\sigma_{3,n}^{2}$ can be modelled by
\vspace{-1.5mm}
\begin{equation}\label{measuresigma}
    \left\{ \begin{aligned}
        \sigma_{1,n}^{2} &= (a_{1}^{2}\sigma^{2})/(P_{\text{A}} N_{\text{sym}} N_{\text{t}} N_{\text{r}} G_{\text{r},n} \sin^{2}{\phi_{n}}),  \\
        \sigma_{i,n}^{2} &= (a_{i}^{2}\sigma^{2})/(P_{\text{A}} N_{\text{sym}} N_{\text{t}} N_{\text{r}} G_{\text{r},n}), i=2,3, 
    \end{aligned}
    \right.
    \vspace{-1.5mm}
\end{equation}
where $N_{\text{sym}}$ is the number of symbols received during each time slot, and $a_{i}, i=1,2,3$ each denotes a constant parameter related to the system configuration \cite{FanLiu2020TWC-PB}. 
Given (\ref{MeasurementModel}) and (\ref{measuresigma}), the radar measurement model can thus be compactly formulated as $\mathbf{y}_{n} = \mathbf{h}(\mathbf{x}_{n}) + \mathbf{z}_{\text{m},n}$, where $\mathbf{h}(\cdot)$ is defined in (\ref{MeasurementModel}) and $\mathbf{z}_{\text{m},n} = [z_{1,n}, z_{2,n}, z_{3,n}]^{T}$ denotes the measurement noise vector with its covariance matrix given by $\mathbf{Q}_{\text{m},n} = \mathrm{diag}(\sigma_{1,n}^{2}, \sigma_{2,n}^{2}, \sigma_{3,n}^{2})$. 
Moreover, the Jacobian matrix for $\mathbf{h}(\cdot)$ can be given by 
\vspace{-1.5mm}
\begin{equation}
	\mathbf{H}_{n} = \begin{bmatrix}
        \iota(\breve{x}_{n}) & \kappa(\breve{x}_{n}) & \zeta(\breve{x}_{n},\breve{v}_{n}) \\
         0 & 0 & \nu(\breve{x}_{n})
    \end{bmatrix}^{T}, 
    \vspace{-1.5mm}
\end{equation}
where the expressions of $\iota(\breve{x}_{n})$, $\kappa(\breve{x}_{n})$, $\zeta(\breve{x}_{n},\breve{v}_{n})$ and $\nu(\breve{x}_{n})$ are given in (\ref{formu::Hne}), respectively.

\subsubsection{Communication Signal Model}
At the $n$th time slot, the downlink LoS channel between the UAV and the device can be expressed as \cite{largescale} $\mathbf{h}^{\text{C}}_{n} = \sqrt{G_{\text{c},n}}e^{-j2\pi\frac{d_{n}}{\lambda}}\mathbf{a}^{H}(\phi_{n})$, where $G_{\text{c},n}=\lambda^{2}/(16\pi^{2}d_{n}^{2})$ is the large-scale channel power gain. 
Then, the ISAC signal received by the device at the $n$th time slot can be represented by $c_{n}(t) = \sqrt{P_{\text{A}}}\mathbf{h}^{\text{C}}_{n}e^{j2\pi\varrho_{n}t}\mathbf{f}_{n}s_{n}(t-\tau_{n}) + z_{\text{c}}(t)$, where $\varrho_{n}$ denotes the Doppler shift observed at the object, and $z_{\text{c}}(t)\sim\mathbf{\mathcal{N}}\left(0, \sigma_{\text{C}}^{2}\right)$ is the additive Gaussian noise. 
Therefore, the achievable date rate can be denoted by $R_{n}=\mathrm{log}_{2}(1+\mathrm{SNR}_{n}^{\text{C}})$, where $\mathrm{SNR}_{n}^{\text{C}}=P_{\text{A}}N_{\text{t}}G_{\text{c},n}/\sigma_{\text{C}}^{2}$ is the receive signal-to-noise ratio (SNR) of the device.

\vspace{-1.5mm}
\subsection{Problem Formulation}
\vspace{-1mm}
In our proposed tracking scheme, the sensing performance is characterized by the weighted sum of predicted PCRBs for state variables \cite{FanLiu2020TWC-PB}. 
Specifically, let $\breve{\mathbf{M}}_{n}$ denote the predicted $n$th time slot's MSE matrix, which can be given by \cite{MKay} $\breve{\mathbf{M}}_{n} = (\breve{\mathbf{M}}_{\text{m},n}^{-1} + \mathbf{M}_{\text{p},n}^{-1})^{-1}$, where $\breve{\mathbf{M}}_{\text{m},n} = (\mathbf{H}_{n}^{H}\breve{\mathbf{Q}}_{\text{m},n}^{-1}\mathbf{H}_{n})^{-1}$ denotes the predicted measurement MSE matrix and $\breve{\mathbf{Q}}_{\text{m},n}=\mathrm{diag}(\breve{\sigma}_{1,n}^{2}, \breve{\sigma}_{2,n}^{2}, \breve{\sigma}_{3,n}^{2})$ denotes the predicted measurement noise covariance matrix. 
The expressions of $\breve{\sigma}_{i,n}^{2}, i=1,2,3$ are given by
\vspace{-1.5mm}
\begin{equation}\label{predMeaNoiVar}
    \!\!\!\!\! \left\{ \begin{aligned}
        &\breve{\sigma}_{1,n}^{2} = a_{1}^{2}\sigma^{2}(H^{2}+\breve{x}_{n}^{2})^{3}/(P_{\text{A}} N_{\text{sym}} N_{\text{t}} N_{\text{r}} \beta_{\text{r}} H^{2}),  \\
        &\breve{\sigma}_{i,n}^{2} = a_{i}^{2}\sigma^{2}(H^{2}+\breve{x}_{n}^{2})^{2}/(P_{\text{A}} N_{\text{sym}} N_{\text{t}} N_{\text{r}} \beta_{\text{r}}), i=2,3.
    \end{aligned}
    \right.
    \vspace{-1mm}
\end{equation}
Thus, the expressions of the predicted PCRBs for the $n$th time slot's state variables can be given by \cite{MKay}
\vspace{-1.5mm}
\begin{equation}\label{formu::PCRB}
    \breve{\mathrm{PCRB}}_{\text{x},n} = \frac{[\breve{\mathbf{M}}_{n}^{-1}]_{22}}{\text{det}(\breve{\mathbf{M}}_{n}^{-1})}, \breve{\mathrm{PCRB}}_{\text{v},n} = \frac{[\breve{\mathbf{M}}_{n}^{-1}]_{11}}{\text{det}(\breve{\mathbf{M}}_{n}^{-1})},
    \vspace{-1mm}
\end{equation}
respectively. 
More specific expressions of $\breve{\mathrm{PCRB}}_{\text{x},n}$ and $\breve{\mathrm{PCRB}}_{\text{v},n}$ are given by (\ref{ABdef}) and (\ref{Ldef}), respectively, where $r_{n,ij}=[\mathbf{M}_{\text{p},n}^{-1}]_{ij},i,j=1,2$.

Without loss of generality, the optimization problem can be formulated as
\vspace{-1.5mm}
\begin{align}
    (\mathrm{P1}): \ &\min_{\breve{\mathbf{x}}_{n}} \ \ \alpha\breve{\mathrm{PCRB}}_{\text{x},n} + (1-\alpha)\breve{\mathrm{PCRB}}_{\text{v},n} \label{opt-obj}\\
    \text{s.t.} \ 
    &\breve{R}_{n} \geq \gamma_{\text{C}}, \tag{\ref{opt-obj}{a}} \label{opt-cstrt-a} \\
    &\left| \breve{x}_{n} - \eta_{n-1} \right| \leq v_{\text{A,max}}\Delta T, \tag{\ref{opt-obj}{b}} \label{opt-cstrt-b} \\
    &\breve{x}_{n} - \Delta T\breve{v}_{n} - \hat{x}_{n-1} = 0, \tag{\ref{opt-obj}{c}} \label{opt-cstrt-c}
    \vspace{-1.5mm}
\end{align}
where $\alpha\in [0,1]$ is the weighting factor, $\breve{R}_{n}=\log_{2}(1+ \frac{P_{\text{A}}\lambda^{2}N_{\text{t}}}{16\pi^{2}(\breve{x}_{n}^{2} + H^2)\sigma_{\text{C}}^{2}} )$ denotes the predicted data rate for the $n$th time slot, $\gamma_{\text{C}}$ represents the communication quality of service (QoS) threshold, $v_{\text{A,max}}$ represents the maximum UAV velocity and $\eta_{n-1}\triangleq\hat{x}_{n-1} + \hat{v}_{n-1}\Delta T + v_{\text{A},n-1}\Delta T$ is derived from the state evolution model. 
In (P1), the predicted QoS constraint (\ref{opt-cstrt-a}) represents that $\breve{R}_{n}$ should be larger than $\gamma_{\text{C}}$ if the UAV reaches the designed trajectory $x_{\text{A},n}$ at the $n$th time slot. 
The inequality (\ref{opt-cstrt-b}) represents the UAV maximum velocity constraint, and the equation (\ref{opt-cstrt-c}) represents the coupling between the predicted state variables.
By solving (P1), the UAV trajectory can be designed as $x_{\text{A},n} = \eta_{n-1} + x_{\text{A},n-1} - \breve{x}_{n}$ with the velocity $v_{\text{A},n} = (x_{\text{A},n} - x_{\text{A},n-1})/\Delta T$. 
However, (P1) is difficult to be optimally solved due to the non-convex objective function.

\begin{figure*}[!t]
    \vspace{-5mm}
    \begin{gather}
        \iota(\breve{x}_{n}) = \frac{H}{H^{2} + \breve{x}_{n}^{2}} , \ \
        \kappa(\breve{x}_{n}) = \frac{2\breve{x}_{n}}{c(H^{2} + \breve{x}_{n}^{2})^{\frac{1}{2}}} , \ \
        \zeta(\breve{x}_{n},\breve{v}_{n}) = \frac{-2f_{c}\breve{v}_{n}H^{2}}{c(H^{2} + \breve{x}_{n}^{2})^{\frac{3}{2}}}, \ \
        \nu(\breve{x}_{n}) = \frac{-2f_{c}\breve{x}_{n}}{c(H^{2} + \breve{x}_{n}^{2})^{\frac{1}{2}}}, \label{formu::Hne} \\
        [\breve{\mathbf{M}}_{n}^{-1}]_{22} = r_{n,22} + \frac{\nu^{2}(\breve{x}_{n})}{\breve{\sigma}_{3,n}^{2}},  \ \
        [\breve{\mathbf{M}}_{n}^{-1}]_{11} = r_{n,11} + \frac{\iota^{2}(\breve{x}_{n})}{\breve{\sigma}_{1,n}^{2}} + \frac{\kappa^{2}(\breve{x}_{n})}{\breve{\sigma}_{2,n}^{2}} + \frac{\zeta^{2}(\breve{x}_{n},\breve{v}_{n})}{\breve{\sigma}_{3,n}^{2}}, \label{ABdef} \\
        \text{det}(\breve{\mathbf{M}}_{n}^{-1}) = [\breve{\mathbf{M}}_{n}^{-1}]_{22}[\breve{\mathbf{M}}_{n}^{-1}]_{11} - (\frac{\nu(\breve{x}_{n})\zeta(\breve{x}_{n},\breve{v}_{n})}{\breve{\sigma}_{3,n}^{2}} + r_{n,12})(\frac{\nu(\breve{x}_{n})\zeta(\breve{x}_{n},\breve{v}_{n})}{\breve{\sigma}_{3,n}^{2}} + r_{n,21}), \label{Ldef} \\
        \rho_{0} = a_1^4 H^4, \rho_{2} = 3\rho_{1}/2 = 288\rho_{0}, \rho_{3} = \rho_{1} + 3\rho_{5}/2, \rho_{4} = 48 \rho_{0} + 17\rho_{5}/4, \rho_{5} =24 a_2^2 c^2 \sqrt{\rho_{0}}, \rho_{6} =-3\rho_{7}=-3 a_{2}^{2}\xi(H), \label{rhocof}
        \vspace{-5mm}
    \end{gather}
    \rule{18.2cm}{0.5pt}
    \vspace{-10mm}
\end{figure*}

\section{Proposed Solution and Approximation} \label{Sec::opt}

\subsection{Proposed Solution to \rm{(P1)}} \label{Subsec::opt-solution}

To efficiently solve (P1), we firstly substitute $(\breve{x}_{n} - \hat{x}_{n-1})/\Delta T$ into $\breve{v}_{n}$ in (P1) according to the constraint (\ref{opt-cstrt-c}). 
As such, the objective function of (P1) can be denoted by $f(\breve{x}_{n})$. 
Then, we apply the SCA method \cite{MMSCA} to obtain a locally optimal solution to (P1). 
Specifically, a sequence denoted by $\{\tilde{x}_{k}^{*}\}$ is iteratively generated to approach the locally optimal solution to (P1). 
The initial point $\tilde{x}_{0}^{*}$ can be set as a feasible solution to (P1). 
In the $k$th iteration, $\tilde{x}_{k}^{*}$ is obtained as the optimal solution to the following optimization problem, i.e.  
\vspace{-1.5mm}
\begin{align}
    (\mathrm{P1.k}): \ &\min_{\tilde{x}_{k}} \ \ \tilde{f}_{k}(\tilde{x}_{k};\tilde{x}_{k-1}^{*}) \label{opt-obj-1k}\\
    \text{s.t.} \ 
    &|\tilde{x}_{k}| \leq x_{\text{c}}, \tag{\ref{opt-obj-1k}{a}} \label{opt-1k-cstrt-a} \\
    &\left| \tilde{x}_{k} - \eta_{n-1} \right| \leq v_{\text{A,max}}\Delta T, \tag{\ref{opt-obj-1k}{b}} \label{opt-1k-cstrt-b}
    \vspace{-6.5mm}
\end{align}  
where the objective function in (\ref{opt-obj-1k}) is defined as $\tilde{f}_{k}(\tilde{x}_{k};\tilde{x}_{k-1}^{*}) \triangleq |f^{(2)}(\tilde{x}_{k-1}^{*})| (\tilde{x}_{k} - \tilde{x}_{k-1}^{*})^{2}/2 + f^{(1)}(\tilde{x}_{k-1}^{*})(\tilde{x}_{k} - \tilde{x}_{k-1}^{*})$, $f^{(i)}(\breve{x}_{n})$ denotes the $i$th order derivate of $f(\breve{x}_{n})$ and $x_{\text{c}} = (\frac{P_{\text{A}}\lambda^{2}N_{\text{t}}}{16\pi^{2}\sigma_{\text{C}}^{2}(2^{\gamma_{\text{C}}}-1)} - H^{2})^{\frac{1}{2}}$ holds due to the constraint (\ref{opt-cstrt-a}). 
(P1.k) is a convex quadratic programming which can be optimally solved by the active set method \cite{ActiveSet}. 
By successively solving (P1.k), a locally optimal solution to (P1) can be found. 

\subsection{Measurement MSE-only Case} \label{Subsec::opt-CRB}

In this subsection, we consider a special case of (P1) to approximate the relative motion state that minimizes the weighted sum of predicted PCRBs. 
Specifically, the state prediction MSE matrix $\mathbf{M}_{\text{p},n}$ is relaxed in our considered case, which practically approximates the case where the predicted measurement MSE for state variables is much lower than the state prediction MSE. 
Moreover, to focus on obtaining the relative motion state that minimizes the weighted sum of predicted PCRBs, the constraints (\ref{opt-cstrt-a})-(\ref{opt-cstrt-c}) are further relaxed, corresponding to the practical case with small communication QoS threshold $\gamma_{\text{C}}$ and large maximum UAV velocity $v_{\text{A,max}}$. 
In this case, the predicted state variables equivalently represent the actual relative position $x_{n}$ and velocity $v_{n}$ between the UAV and the object, and the predicted PCRBs for state variables are correspondingly reduced as the measurement CRBs, denoted by $\mathrm{CRB}_{\text{x},n}(x_{n})$ and $\mathrm{CRB}_{\text{v},n}(x_{n},v_{n})$, respectively.
Specifically, $\mathrm{CRB}_{\text{x},n}(x_{n})$ and $\mathrm{CRB}_{\text{v},n}(x_{n},v_{n})$ can be obtained as the diagonal elements of $\breve{\mathbf{M}}_{\text{m},n}$ with $\breve{\mathbf{x}}_{n}$ substituted by $\mathbf{x}_{n}$ and further expressed as $\mathrm{CRB}_{\text{x},n}(x_{n}) = (\frac{\iota^{2}(x_{n})}{ \sigma_{1,n}^{2}} + \frac{\kappa^{2}(x_{n})}{ \sigma_{2,n}^{2}} )^{-1}$ and $\mathrm{CRB}_{\text{v},n}(x_{n},v_{n}) = (\frac{\nu^{2}(x_{n})}{ \sigma_{3,n}^{2}})^{-1} + \frac{\zeta^{2}(x_{n},v_{n})}{\nu^{2}(x_{n})}(\frac{\iota^{2}(x_{n})}{ \sigma_{1,n}^{2}} + \frac{\kappa^{2}(x_{n})}{ \sigma_{2,n}^{2}})^{-1}$, respectively.
Then, the considered special case of (P1) can be formulated as\footnote{For notational simplicity, we only consider the case with $x_{n} \geq 0$ in (SP1) since $g(x_{n},v_{n})$ is symmetric between the cases where $x_{n}>0$ and $x_{n}<0$.}
\vspace{-1.5mm}
\begin{equation}
    (\mathrm{SP1}): \ \min_{\mathbf{x}_{n}} \ \ g(x_{n},v_{n}) \ \ \text{s.t.} \ x_{n} \geq 0, \label{opt-sp}
    \vspace{-1.5mm}
\end{equation}
where $g(x_{n},v_{n}) = \alpha\mathrm{CRB}_{\text{x},n}(x_{n}) + (1-\alpha)\mathrm{CRB}_{\text{v},n}(x_{n},v_{n})$ denotes the objective function of (SP1) which is generally non-convex for $x_{n}\in[0,\infty)$.

To optimally solve (SP1), we firstly find the optimal relative velocity $v_{n}^{*}$. 
For any given feasible $x_{n}$, $\frac{\partial^{2} g(x_{n},v_{n})}{\partial v_{n}^{2}}$ is nonnegative, thus $v_{n}^{*}$ can be given by the solution to the equation $\frac{\partial g(x_{n},v_{n}^{*})}{\partial v_{n}^{*}} = 0$ derived from the Karush-Kuhn-Tucker (KKT) conditions \cite{CVX}, i.e., $v_{n}^{*}=0$. 
As a result, minimizing $g(x_{n},v_{n})$ is equivalent to minimizing $g(x_{n},0)$. 
Then, a sufficient condition for $g(x_{n},0)$ being convex is given as follows.
\begin{lemma}\label{lem1}
    $g(x_{n},0)$ is convex for $x_{n}\in [x_{l}, \infty)$ with
    \vspace{-1.5mm}
    \begin{equation}
        x_{l} = \left\{ \begin{aligned}
            &0, \ \ \xi(H) \leq 0,  \\
            &H/\sqrt{\overline{\chi}}, \ \ \xi(H) > 0, 
        \end{aligned}
        \right.
        \vspace{-1.5mm}
    \end{equation}
    with $\overline{\chi} = 2(1+(5c^{2}a_{2}^{2}/\xi(H)))^{\frac{1}{2}} $ and $\xi(H) = 4 a_{1}^{2} H^{2} - 5 c a_{2}^{2}$.
\end{lemma}
\begin{proof}
    We firstly analyze the convexity of $\mathrm{CRB}_{\text{x},n}(x_{n})$ and $\mathrm{CRB}_{\text{v},n}(x_{n},0)$ for $x_{n}\in(0,\infty)$, respectively. 
    Specifically, the convexity of $\mathrm{CRB}_{\text{x},n}(x_{n})$ for positive $x_{n}$ depends on the sign of $\xi(H)$.
    Let us define a variable $\chi\triangleq(H/x_{n})^{2}$ and denote the $i$th order derivate of $\mathrm{CRB}_{\text{x},n}(x_{n})$ and $\mathrm{CRB}_{\text{v},n}(x_{n},0)$ by $\mathrm{CRB}^{(i)}_{\text{x},n}(x_{n})$ and $\mathrm{CRB}^{(i)}_{\text{v},n}(x_{n},0)$, respectively. 
    Then, by substituting $\chi$ into $x_{n}$, $\mathrm{CRB}^{(2)}_{\text{x},n}(\chi)$ can be expressed as
    \vspace{-2.5mm}
    \begin{equation} 
        \mathrm{CRB}^{(2)}_{\text{x},n}(\chi) = \frac{B (\chi+1)^3 (\sum_{j=1}^{7}\rho_{j}\chi^{j})}{\chi^4 (a_2^2 c^2 \chi^3+4 a_1^2 H^2 (\chi+1)^2)^3}, 
        \vspace{-1.5mm}
    \end{equation}
    where the expressions of $B$ and $\rho_{j},j=1,..,7$ are given by $B = \frac{2 a_1^2 a_2^2 c^2 H^6 \sigma^2}{N_{\text{r}} N_{\text{t}} P_A \beta_r N_{\text{sym}}}$ and (\ref{rhocof}), respectively.
    When $\xi(H)$ is nonpositive, $\rho_{j}\geq 0 ,j=1,..,7$ hold and thus $\mathrm{CRB}^{(2)}_{\text{x},n}(\chi)$ is positive for $\chi\in(0,\infty)$, which equivalently indicates that $\mathrm{CRB}_{\text{x},n}(x_{n})$ is convex for positive $x_{n}$. 
    When $\xi(H)$ is positive, both $\rho_{6}$ and $\rho_{7}$ are negative. 
    Thus, $\mathrm{CRB}^{(2)}_{\text{x},n}(\chi)$ has at most one zero for positive $\chi$ due to the Descartes' rule of signs. 
    However, since $\mathrm{CRB}^{(2)}_{\text{x},n}(\overline{\chi})$ is positive but $\lim\limits_{\chi\rightarrow\infty}\mathrm{CRB}^{(2)}_{\text{x},n}(\chi)$ is negative, $\mathrm{CRB}^{(2)}_{\text{x},n}(\chi)$ has at least one zero in the range given by $(\overline{\chi},\infty)$. 
    Therefore, $\mathrm{CRB}^{(2)}_{\text{x},n}(\chi)$ stays positive for $\chi\in(0,\overline{\chi}]$, which is equivalent to that $\mathrm{CRB}_{\text{x},n}(x_{n})$ is convex for $x_{n}\in[x_{l}, \infty)$. 
    Moreover, $\mathrm{CRB}_{\text{v},n}(x_{n},0)$ is convex for positive $x_{n}$ since $\mathrm{CRB}^{(2)}_{\text{v},n}(x_{n},0)$ is positive for $x_{n}\in(0,\infty)$. 
    Since $g(x_{n},0)$ is the nonnegative weighted sum of $\mathrm{CRB}_{\text{x},n}(x_{n})$ and $\mathrm{CRB}_{\text{v},n}(x_{n},0)$, it is convex for $x_{n}\in[x_{l}, \infty)$, which completes the proof.
\end{proof}

Given Lemma \ref{lem1}, we provide the optimal relative position $x_{n}^{*}$ in the following proposition. 
\begin{proposition}\label{pro1}
    The optimal relative position $x_{n}^{*}$ is unique and given by the stationary point of $g(x_{n},0)$ located in the range $[x_{l}, x_{u}]$ with $x_{u}=H/\sqrt{2}$. 
\end{proposition}
\begin{proof}
    We firstly discuss the following three cases: $\alpha=0$, $\alpha=1$ and $\alpha\in(0,1)$. 
    In the case with $\alpha=0$, $g(x_{n},0)=\mathrm{CRB}_{\text{v},n}(x_{n},0)$ is minimized at its stationary point $x_{n}^{*}=x_{u}$. 
    In the case with $\alpha=1$, $g(x_{n},0)=\mathrm{CRB}_{\text{x},n}(x_{n})$ is also minimized at its stationary point, which is given by $x_{n}^{*}=0$ when $\xi(H)$ is nonpositive and $x_{n}^{*}=H/\sqrt{\chi_{1}}>x_{l}$ with $\chi_{1}=\overline{\chi}\cos{(\frac{1}{3}\arctan{\frac{\sqrt{5}c a_{2}}{\sqrt{\xi(H)}}})}$ when $\xi(H)$ is positive. 
    In the case with $\alpha\in(0,1)$, we then discuss the following two subcases: nonpositive $\xi(H)$ and positive $\xi(H)$. 
    When $\xi(H)$ is nonpositive, $g^{(1)}(x_{n},0)$ is a monototically increasing function for positive $x_{n}$ since $g(x_{n},0)$ is convex for positive $x_{n}$ according to Lemma \ref{lem1}. 
    Then, because $g^{(1)}(0,0)g^{(1)}(x_{u},0)<0$ is satisfied, $g(x_{n},0)$ has and only has one stationary point located in the range $[x_{l}, x_{u}]$. 
    Again, since $g(x_{n},0)$ is convex for positive $x_{n}$, the optimal relative position $x_{n}^{*}$ is just given by the stationary point. 
    When $\xi(H)$ is positive, $g^{(1)}(x_{n},0)<0$ holds for $x_{n}\in(0, x_{l})$. 
    Since $g^{(1)}(x_{u},0)>0$ is still satisfied and $g(x_{n},0)$ is convex for $x_{n}\in[x_{l}, \infty)$, the optimal relative position $x_{n}^{*}$ is still given by the unique stationary point located in the range $[x_{l}, x_{u}]$, which completes the proof. 
\end{proof}

The stationary point in Proposition \ref{pro1} can be obtained using the Newton's method. 
Furthermore, the unique optimal relative position $x_{n}^{*}$ and the zero optimal relative velocity $v_{n}^{*}=0$ indicates that the optimal relative motion state between the UAV and the object, denoted by $\mathbf{x}_{n}^{*}=[x_{n}^{*}, 0]^{T}$, is a relatively static state where the UAV keeps a fixed elevation angle $\phi_{n}^{*} = \arctan{(H/x_{n}^{*})}$ relative to the object.

\begin{figure*}[htbp]
	\centering
	\vspace{-4mm}
	\subfigure[The UAV and object motion state.]{
		\label{simufig1}
		\vspace{-1mm}
		\includegraphics[width=5.5cm]{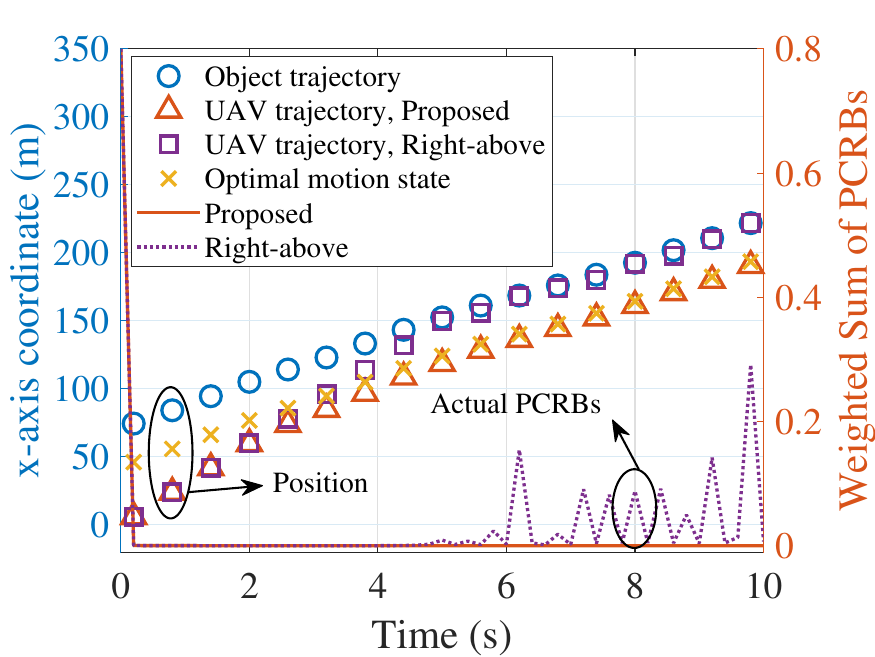}
	}
	\subfigure[The fixed elevation angle.]{
		\label{simufig2}
        \vspace{-1mm}
		\includegraphics[width=5.5cm]{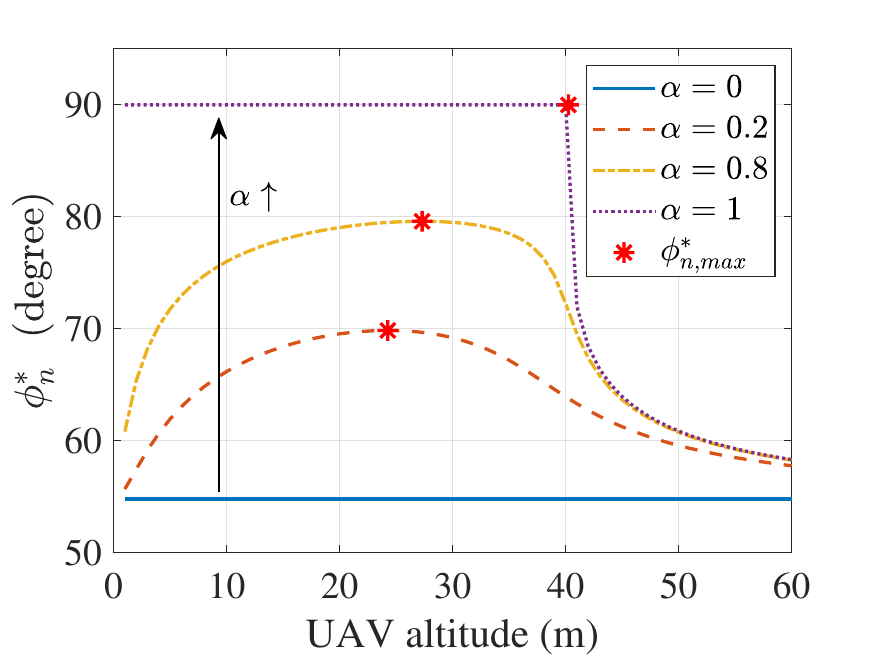}
	}
	\subfigure[The S\&C performance bounds.]{
		\label{simufig3}
        \vspace{-1mm}
		\includegraphics[width=5.5cm]{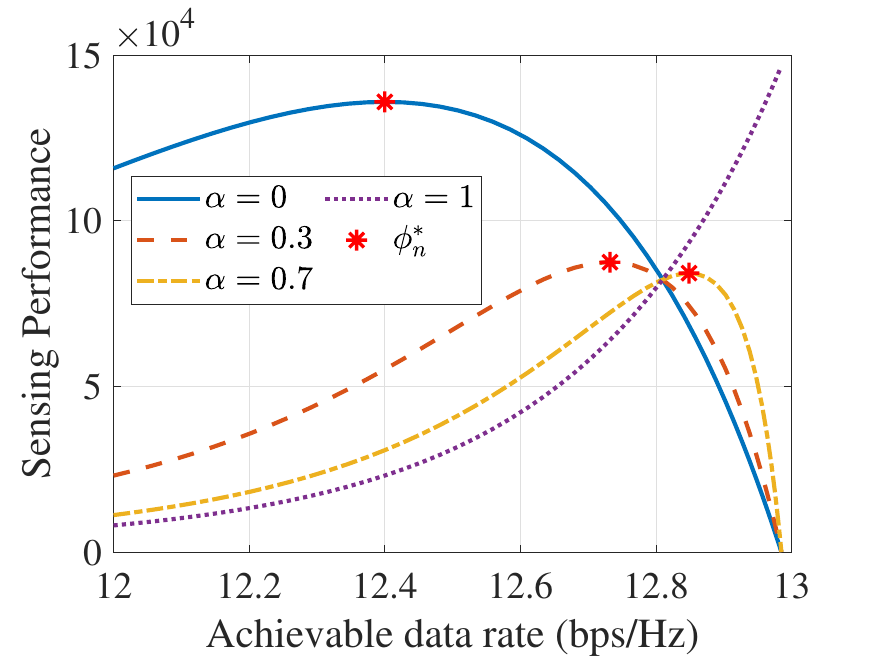}
	}
    \vspace{-2mm}
	\caption {Numerical results characterizing the proposed tracking scheme, the fixed elevation angle and the sensing\&communication (S\&C) performance bounds. }
	\label{figureSimu}
\end{figure*}

\section{Numerical results} \label{Sec::Num}

In this section, numerical results are provided for characterizing the solution in Section \ref{Subsec::opt-solution} and the optimal relative motion state derived in Section \ref{Subsec::opt-CRB}. 
Unless specified otherwise, the following required parameters are used: $P_{\text{A}} = 40$ dBm, $N_{\text{sym}} = 10^{4}$, $\Delta T = 0.2$ s, $\lambda = 0.01$ m, $f_{\text{c}} = 30$ GHz, $\sigma^{2} = \sigma_{\text{C}}^{2} = -80$ dBm, $\gamma_{\text{C}} = 11$ bps/Hz, $\tilde{q} = 5$, $\varepsilon = 100 \ \text{m}^{2}$ , $N_{\text{t}}=N_{\text{r}}=32$, $a_{1} = 1$, $a_{2} = 1.2\times 10^{-7}$, $a_{3} = 600$, and $H=50$ m. 
The proposed scheme is compared with the benchmark ``Right-above'': the UAV tends to be right above the moving object if (\ref{opt-cstrt-b}) is satisfied, or otherwise be as close to the moving object as possible.

Fig.\ref{simufig1} shows the trajectories of both the UAV and the moving object together with the weighted sum of actual PCRBs under the proposed tracking scheme and the benchmark in one trial, where the results are picked per $0.6$s and the weighting factor is set as $\alpha=0.5$.
Under this setup, the trace of the state prediction MSE matrix $\mathrm{Tr(\mathbf{M}_{\text{p},n})}$ and the trace of the predicted measurement MSE matrix $\mathrm{Tr(\breve{\mathbf{M}}_{\text{m},n})}$ at a typical time slot (i.e., $n=40$) are given by $\mathrm{Tr(\mathbf{M}_{\text{p},n})} = 1.014$ and $\mathrm{Tr(\breve{\mathbf{M}}_{\text{m},n})} = 2.38\times10^{-4} $, respectively, which represents the case where the predicted measurement MSE dominates the predicted PCRBs. 
Then, it can be easily observed that the UAV finally approximates its motion state as the optimal relative motion state $\mathbf{x}_{n}^{*}$, which validates the effectiveness of approximating the relative motion state that minimizes the weighted sum of predicted PCRBs by $\mathbf{x}_{n}^{*}$. 
Also, the effectiveness of our proposed tracking scheme is verified since it keeps a relatively low weighted sum of actual PCRBs. 
Moreover, our proposed tracking scheme outperforms the ``Right-above" scheme, which generates both the fluctuated UAV trajectory and the PCRBs.

Fig.\ref{simufig2} illustrates the varying trends of the fixed elevation angle $\phi_{n}^{*}$ as the UAV flying altitude $H$ increases under different weighting factors. 
It can be observed that $\phi_{n}^{*}$ in general first increases but then decreases. 
This is expected due to the \textbf{measurement trade-off:} $\iota^{2}(x_{n})/\sigma_{1,n}^{2}$ is maximized at $\phi_{n}=\pi/2$, however, $\kappa^{2}(x_{n})/\sigma_{2,n}^{2}$ and $\nu^{2}(x_{n})/\sigma_{3,n}^{2}$ are both maximized at $\phi_{n}=\arctan{\sqrt{2}}$. 
Let $\tilde{H}$ denote the UAV altitude that reaches $\phi_{n,\text{max}}^{*}$. 
When $H<\tilde{H}$ holds, the role of maximizing $\iota^{2}(x_{n})/\sigma_{1,n}^{2}$ on the measurement CRB minimization outweighs the counterpart of maximizing $\kappa^{2}(x_{n})/\sigma_{2,n}^{2}$ and $\nu^{2}(x_{n})/\sigma_{3,n}^{2}$, while the reverse holds when $H>\tilde{H}$ is satisfied. 
Particularly, $\phi_{n}^{*}$ stays constant as $\arctan{\sqrt{2}}$ with $\alpha=0$ because only $\nu^{2}(x_{n})/\sigma_{3,n}^{2}$ is needed to be maximized. 
In contrast, $\phi_{n}^{*}$ sharply decreases with $\alpha=1$ since $H>\tilde{H}\approx 40.25 \ \text{m}$ results in a sudden change between the importance of maximizing $\iota^{2}(x_{n})/\sigma_{1,n}^{2}$ and maximizing $\kappa^{2}(x_{n})/\sigma_{2,n}^{2}$ on minimizing the measurement CRB.
Furthermore, the increasing trend of $\phi_{n}^{*}$ as $\alpha$ increases from $0$ to $1$ under the same $H$ illustrates the \textbf{position-velocity trade-off:} the measurement CRB for the relative position is minimized at $\phi_{n}=\pi/2$ when $\xi(H)$ is nonpositive and at $\phi_{n}=\arctan{\sqrt{\chi_{1}}}$ when $\xi(H)$ is positive, however, the counterpart for relative velocity is minimized at $\phi_{n}=\arctan{\sqrt{2}}$.

Fig.\ref{simufig3} illustrates the varying trends of the sensing performance with the increasing of the communication performance with $a_{1}=0.15$, where the sensing performance is calculated as the inverse of $g(x_{n},0)$ and the communication performance is calculated as the achievable date rate $R_{n}$.
As the weighting factor $\alpha$ increases from $0$ to $1$, the sensing-communication performance bound successively approaches the performance bound with $\alpha=1$, which demonstrates the \textbf{sensing-communication trade-off:} the weighted sum of measurement CRBs is minimized at $\phi_{n}=\phi_{n}^{*}$ is satisfied, however, the achievable data rate is maximized at $\phi_{n}=\pi/2$.
Moreover, the sensing performance approaches zero with $\alpha \in [0,1)$ when the achievable data rate is maximized because $\phi_{n}=\pi/2$ results in infinite measurement CRB for the estimated relative velocity.

\section{Conclusion}

In this work, an EKF-based tracking scheme is proposed for a UAV tracking a moving object, where an efficient solution to the weighted sum-predicted PCRB minimization is obtained based on the SCA method. 
Under the predicted measurement MSE-dominant case, the solution can be approximated by the optimal relative motion state obtained under the measurement MSE-only case, which is proved to sustain a fixed elevation angle and zero relative velocity. 
Furthermore, simulation results validate the effectiveness of the proposed tracking scheme and the approximation as well as three interesting trade-offs on system performance achieved by the fixed elevation angle. 
More general scenarios where multiple objects moving along trajectories of arbitrary geometry are worthwhile future works. 


\bibliographystyle{IEEEtran}
\bibliography{IEEEabrv,ref}

\end{document}